\newif\ifAPOCS
\newif\ifNoFormat

\APOCStrue
\NoFormatfalse

\ifAPOCS
\documentclass[twoside,leqno]{article}

\usepackage[letterpaper]{geometry}

\usepackage{SIAM/ltexpprt}
\usepackage{hyperref}

\title{\Large Improved Bounds for Scheduling Flows under Endpoint Capacity Constraints}
\author{
Searidang Pa\thanks{Northeastern University, \href{mailto:pa.s@northeastern.edu}{pa.s@northeastern.edu}}
\and Rajmohan Rajaraman \thanks{Northeastern University, \href{mailto:rajaraman.r@northeastern.edu }{rajaraman.r@northeastern.edu}}
\and David Stalfa\thanks{Northeastern University, \href{mailto:stalfa.d@northeastern.edu}{stalfa.d@northeastern.edu}.}
}

\date{}

\usepackage{mathtools}
\usepackage{tabularx, colortbl} % for flexible tables and coloring tables
\usepackage{xcolor}
\definecolor{lightgray}{rgb}{0.9, 0.9, 0.9}

\usepackage{graphicx}
\usepackage{tikz}
\usepackage{amsmath}
\usepackage{amssymb} % clashes with sigconf
\usepackage[inline]{enumitem}

\usepackage[algo2e,vlined,ruled,linesnumbered]{algorithm2e}
\SetEndCharOfAlgoLine{}
\SetKwInput{KwInit}{Initialize}

\DeclarePairedDelimiter\ceil{\lceil}{\rceil}
\DeclarePairedDelimiter\floor{\lfloor}{\rfloor}

\makeatletter
\newcommand*{\rom}[1]{\expandafter\@slowromancap\romannumeral #1@}
\makeatother

\newcommand{\ecfs}{\mbox{\sc ecfs}}

\newtheorem{thm}{Theorem}

\newcommand{\rowspace}{\textcolor{white}{\raisebox{-1.5mm}{\scalebox{3}{$($}}}}
\newcommand{\toprowspace}{\textcolor{lightgray}{\raisebox{-.5mm}{\scalebox{2}{$($}}}}
\newcommand{\opt}{\mbox{\textsc{Opt}}}
\newcommand{\alg}{\mbox{\textsc{Alg}}}

\newcommand{\aster}{\textasteriskcentered{}}

\newcommand{\junk}[1]{}

\fi

\ifNoFormat
\documentclass[sigconf, 10pt]{acmart}
\usepackage[utf8]{inputenc}

\renewcommand\footnotetextcopyrightpermission[1]{} % removes footnote with conference information in first column
\usepackage{balance}

\title{Improved Bounds for Scheduling Flows under Endpoint Capacity Constraints}
\author{Searidang Pa \ \ \ \ Rajmohan Rajaraman \ \ \ \ David Stalfa}

\authornote{
Khoury College of Computer Sciences, Northeastern University, Boston MA 02115.  Email: {\tt \{pa.s,r.rajaraman,stalfa.d\}@northeastern.edu}}

\usepackage[algo2e,vlined,ruled,linesnumbered]{algorithm2e}
\SetEndCharOfAlgoLine{}
\SetKwInput{KwInit}{Initialize}

\newtheorem{theorem}{Theorem}
\newtheorem{corollary}{Corollary}[theorem]
\newtheorem{lemma}{Lemma}[section]

\usepackage{mathtools}
\usepackage{tabularx} % for flexible tables
\usepackage{graphicx}
\usepackage{hyperref}
\usepackage{tikz}
\usepackage{amsmath}
\usepackage{amsthm}
\usepackage[inline]{enumitem}

\DeclarePairedDelimiter\ceil{\lceil}{\rceil}
\DeclarePairedDelimiter\floor{\lfloor}{\rfloor}

\makeatletter
\newcommand*{\rom}[1]{\expandafter\@slowromancap\romannumeral #1@}
\makeatother

\newcommand{\ecfs}{\mbox{\sc ecfs}}

\newcommand{\rowspace}{\textcolor{white}{$\displaystyle \frac{\frac{1}{1}}{\frac{1}{1}}$}}
\newcommand{\toprowspace}{\textcolor{white}{$\displaystyle \frac{1}{1}$}}

\usepackage{thmtools, thm-restate}
\declaretheoremstyle[
spaceabove=6pt, spacebelow=6pt,
headfont=\normalfont\bfseries,
headindent=0em,
bodyfont=\normalfont\itshape,
postheadspace=0.5em,
]{mystyle}
\declaretheorem[name=Theorem, style=mystyle]{thm} %restatable theorem
\theoremstyle{mystyle} %keeps theorem styles consistent

\renewcommand{\comment}[1]{\textcolor{blue}{#1}}
\newcommand{\junk}[1]{}

\begin{abstract}
We consider the problem of scheduling flows subject to endpoint capacity constraints.  We are given a set of capacitated nodes and an online sequence of requests where each request has a release time and a demand that needs to be routed between two nodes.  A schedule specifies, for each time step, the requests that are routed in that step under the constraint that the total demand routed on a node in any step is at most its capacity.  A key performance metric in such a scheduling scenario is the response time (or flow time) of a request, which is the difference between the time the request is completed and its release time.  Previous work has shown that it is impossible to achieve bounded competitive ratio for average response time without resource augmentation, and that a constant factor competitive ratio is achievable with augmentation exceeding two (Dinitz-Moseley Infocom 2020).  For the maximum response time objective, the best known result is a 2-competitive algorithm with a resource augmentation at least 4 (Jahanjou et al SPAA 2020).

In this paper, we present improved bounds for the above flow scheduling problem under various response time objectives.  Our first result is a lower bound showing that, without resource augmentation, the best competitive ratio for the maximum response time objective is $\Omega(n)$, where $n$ is the number of nodes. The remaining results present simple, resource-augmented algorithms that are competitive for the maximum response time objective in their respective settings. Our first algorithm, Proportional Allocation, uses $(1+\varepsilon)$ resource augmentation to achieve a  $(1/\varepsilon)$-competitive ratio for maximum response time in the setting with general demands, general capacities, and splittable jobs, for any $\varepsilon > 0$. Our second algorithm, Batch Decomposition, is $2$-competitive (resp., matches optimum) for maximum response time using resource augmentation 2 (resp., 4) in the setting with unit demands, unit capacities, and unsplittable jobs.  We also derive bounds for the simultaneous approximation of average and maximum response time metrics.
\end{abstract}

\fi

\begin{document}

\maketitle

\ifNoFormat
\fancyhead[]{}      % removes header
\fi

\ifAPOCS
\begin{abstract} 
\small\baselineskip=9pt 

\end{abstract}
\fi 

% rajmohan
\section{Introduction}
We study the problem of scheduling an online sequence of flows subject to endpoint capacity constraints, which we call \emph{Endpoint Capacitated Flow Scheduling} or \ecfs.  Such flow scheduling instances arise in diverse scenarios in networks and distributed systems, where requests for communication or data transfer need to be satisfied, while meeting constraints at the source and destination ports of the requests.  For instance, allocating resources for packet flows in a crossbar switch is a classic example of this flow scheduling problem~\cite{Giaccone26,Gong28,Guez,Shah50}.  More recently, variants of \ecfs\ have been used to model the scheduling of dynamically reconfigurable  topologies~\cite{jia+etal.opticalwan.17,dinitz+moseley.reconfigurable.20}.

We model \ecfs\ by a set of capacitated nodes and an online sequence of requests where each request has a release time and a demand that needs to be routed between two nodes.  An \ecfs\ \emph{schedule} specifies, for each time step, the requests that are allocated in that step under the constraint that the total demand allocated at any endpoint in any step is at most its capacity.  For a given request, two measures of interest are its \emph{completion time} and its \emph{response time}, which is the difference between its completion time and its release time.  \junk{While extensive work has been done on optimizing the makespan (maximum completion time) and average completion time~\cite{}, response time metrics have only recently gained attention.}  In online scheduling scenarios with variable release times of flows, response time metrics are arguably of greater significance~\cite{dinitz+moseley.reconfigurable.20}.  

\subsection{Overview of Results \& Techniques.}
The focus of this paper is on designing and analyzing online algorithms for \ecfs\ under response time metrics in the competitive analysis framework~\cite{borodin+e:online,sleator+t:online}.  Two natural variants of response time metrics are the \emph{average response time}, which considers the overall performance of the schedule by giving equal weight to each job, and the \emph{maximum response time}, which is motivated by fairness and aims to provide a guarantee to all jobs. We differentiate between the \emph{general case}, where demands and capacities are arbitrary integers, and the \emph{unit case}, where demands and capacities are unit.  We further differentiate between \emph{splittable} jobs, which can be fractionally scheduled over multiple time steps, and \emph{unsplittable} jobs which need to be satisfied in a single step.  

Previous work has shown that it is impossible to achieve a bounded competitive ratio for average response time without resource augmentation~\cite{dinitz+moseley.reconfigurable.20,jahanjou+rajaraman+stalfa.flowswitch.20} and that a constant factor competitive ratio is achievable for splittable jobs with augmentation exceeding two~\cite{dinitz+moseley.reconfigurable.20}.  For the maximum response time objective, the best known result is a 2-competitive algorithm with a resource augmentation 4~\cite{jahanjou+rajaraman+stalfa.flowswitch.20}.  In this paper, we present \emph{improved bounds for \ecfs\ under various response time objectives}. 

\paragraph{Lower bound for maximum response time.}  Our first result is a lower bound on the best competitive ratio achievable for maximum response time (without resource augmentation).  

\ifAPOCS
\begin{thm}
For any deterministic online \ecfs\ algorithm Alg that is allowed to split jobs, there exists an instance with unit demands and capacities such that the maximum response time of Alg is $\Omega(n)$ times that achieved by an optimal schedule that does not split jobs, where $n$ is the number of nodes.
\label{thm:lower}
\end{thm}
\fi

\ifNoFormat
\begin{restatable}{thm}{lowerbound}
For any deterministic online \ecfs\ algorithm $Alg$ that is allowed to split jobs, there exists an instance with unit demands and capacities such that the maximum response time of $Alg$ is $\Omega(n)$ times that achieved by an optimal schedule that does not split jobs, where $n$ is the number of nodes.
\label{thm:lower}
\end{restatable}
\fi

The strong lower bound of Theorem~\ref{thm:lower} holds even for bipartite flows. The proof is an adaptation of an elegant construction in \cite{irani+leung.conflicts.03} which establishes a similar lower bound for \emph{conflict scheduling}, a generalization of \ecfs\ with unsplittable jobs. Together with previous lower bounds for average response time, Theorem~\ref{thm:lower} indicates that competitive ratio bounds for \ecfs\ under response time metrics are more meaningful in settings that allow for resource augmentation.  In the remainder of the paper, we focus on the tradeoff between resource augmentation and the competitive ratio achieved by \ecfs\ schedules.

\paragraph{Upper bounds for maximum response time.} 
Our next set of main results provides upper bounds on the competitiveness of two algorithms. The first, Proportional Allocation, applies in the setting with general demands, general capacities, and splittable jobs. The second, Batch Decomposition, applies in the setting with unit demands, unit capacities, and unsplittable jobs.

\ifAPOCS
\begin{thm}
Given an \ecfs\ instance with splittable jobs, Proportional Allocation yields a $(1/\varepsilon)$-competitive schedule for maximum response time, assuming the capacity of each node is increased by a factor of $1 + \varepsilon$, for any $\varepsilon > 0$.
\label{thm:propalloc}
\end{thm}

\ifNoFormat
\begin{restatable}{thm}{propalloc}
For any \ecfs\ instance with splittable jobs, the Proportional Allocation algorithm yields a schedule that is $(1/\varepsilon)$-competitive for maximum response time, assuming the capacity of each node is increased by a factor of $1 + \varepsilon$, for any $\varepsilon > 0$.
\label{thm:propalloc}
\end{restatable}
\fi

\ifAPOCS
\begin{thm}
Given an \ecfs\ instance with unsplittable jobs, unit demands, and unit capacities, Batch Decomposition yields a $\max\{2/k,1\}$-competitive schedule for maximum response time, assuming the capacity of each node is increased by a factor of $2k$, for $k \in \{1,2\}$. 
\label{thm:batch}
\end{thm}
\fi

\ifNoFormat
\begin{restatable}{thm}{batch}
Consider a given \ecfs\ instance with unsplittable jobs, unit demands and capacities. Then, by increasing all node capacities by a factor of $2k$ for positive integer $k$, Batch Decomposition schedules all requests with response time 
at most $\max\{(2/k), 1\}$ times the optimal maximum response time. 
\label{thm:batch}
\end{restatable}
\fi 

\junk{
\begin{theorem}
Consider all $p$-norms over response times. There exists an algorithm that is $2(2+ \varepsilon)/\varepsilon$-competitive simultaneously for all $p \in (1,\infty)$, where the algorithm increases the capacity of each port by a factor of $2 + \varepsilon$. \comment{Dang is working on this}
\end{theorem}
}

\begin{table*}
{%
\centering
\begin{tabular}{|lccccc|}
    \hline
    \rowcolor{lightgray}
    Algorithm \toprowspace & Resources & Competitive Ratio & Objective & Setting & Source   \\
    \hline &&&&& \vspace{-3.5mm} \\ 
    \hline
    Proportional \rowspace & $1 + \varepsilon$  & $1/\varepsilon + 1/\opt$ &  maximum & {\small\begin{tabular}{@{}c@{}}splittable jobs\\  general capacities/demands\end{tabular}} & \aster \\
    \hline 
    LP Batch \rowspace & $4$ &  $2$  & maximum  & {\small\begin{tabular}{@{}c@{}}unsplittable jobs\\  general capacities/demands\end{tabular}} & \cite{jahanjou+rajaraman+stalfa.flowswitch.20} \\
    \hline
    Batch Decomposition \rowspace & $2k$   & $\max\{2/k, 1\}$ & maximum & {\small\begin{tabular}{@{}c@{}}unsplittable jobs\\unit capacities/demands\end{tabular}} & \aster \\
    \hline
    FIFO \rowspace & $2+k$  &  $\max\{2/k, 1\}$ & maximum  & {\small\begin{tabular}{@{}c@{}}unsplittable jobs\\unit capacities/demands\end{tabular}} & \aster   \\
    \hline
    FIFO \rowspace & $2+k$ & $2(2+k)/k$   & average &{\small\begin{tabular}{@{}c@{}}unsplittable jobs\\unit capacities/demands\end{tabular}} & \cite{dinitz+moseley.reconfigurable.20}  \\
    \hline
    Shortest Job First \rowspace & $2 + \varepsilon$   & $2(2+\varepsilon) / \varepsilon$  & average  & {\small\begin{tabular}{@{}c@{}}splittable jobs\\  general capacities/demands\end{tabular}} & \cite{dinitz+moseley.reconfigurable.20} \vspace{-.5mm}
    %\\     \hline  \vspace{-4mm} &&&&&
\end{tabular}
}\\%
\begin{tabular}{|l|r|}
    \hline
    \rowcolor{lightgray}
    {\small $\varepsilon$ is any positive real number, $k$ is any positive integer} \hspace {13.3mm} & \hspace{13.22mm} \textasteriskcentered{} {\small indicates a result proved in this paper} \\ 
    \hline
\end{tabular}
\caption{Comparison of Algorithms for Endpoint Capacitated Flow Scheduling}
\label{tab:bounds}
\end{table*}

\junk{

\begin{table*}
{%
\centering
\begin{tabular}{|ll@{}lccc|}
    \hline
    Algorithm \toprowspace & Capacity I&ncrease & Objective & Competitive Ratio  & Setting   \\
    \hline \vspace{-3.5mm} \\ \hline
    LP Batch \rowspace & $4$ && maximum & $2$ \cite{jahanjou+rajaraman+stalfa.flowswitch.20} & {\small\begin{tabular}{@{}c@{}}unsplittable jobs\\  general capacities/demands\end{tabular}} \\
    \hline
    Proportional \rowspace & $1 + \varepsilon$ & $\varepsilon > 0$ & maximum & $1/\varepsilon + 1/OPT$ * & {\small\begin{tabular}{@{}c@{}}splittable jobs\\  general capacities/demands\end{tabular}}  \\
    \\[-0.8em] \hline \\[-0.8em]
    Shortest Job First & $2 + \varepsilon$ & $\varepsilon > 0$  & average & $2(2+\varepsilon) / \varepsilon$ \cite{dinitz+moseley.reconfigurable.20}  & {\small\begin{tabular}{@{}c@{}}splittable jobs\\  general capacities/demands\end{tabular}} \\
    \\[-0.8em] \hline \\[-0.8em]
    FIFO & $2+k$ & $k \in \mathbb{N}^+$  & \begin{tabular}{@{}c@{}}average\\ \\[-0.8em] maximum\end{tabular}&  \begin{tabular}{@{}c@{}}2(2+$k$)/$k$ \cite{dinitz+moseley.reconfigurable.20} \\ \\[-0.8em] $\max\{2/k, 1\}$ *\end{tabular} &{\small\begin{tabular}{@{}c@{}}unsplittable jobs\\unit capacities/demands\end{tabular}}   \\
    \\[-0.8em] \hline \\[-0.8em]
    Batch Decomposition & $2k$ & $k \in \mathbb{N}^+$ & maximum & $\max\{2/k, 1\}$ * & {\small\begin{tabular}{@{}c@{}}unsplittable jobs\\unit capacities/demands\end{tabular}}    \\
    \\[-0.8em] \hline %\\[-0.8em]
\end{tabular}
}\\%
\begin{tabular}{|l|r|}
    \hline
    {\small $\varepsilon$ is any positive real number greater, $k$ is any positive integer} & \hspace{6mm} * indicates results proved in this paper \\ 
    \hline
\end{tabular}
\begin{flushright}
* indicates results proved in this paper
\end{flushright}
\caption{Comparison of Algorithms for Endpoint Capacitated Flow Scheduling}
\label{tab:bounds}
\end{table*}

\medskip
\begin{center}
\begin{tabular}{lccc}
    &\begin{tabular}{@{}c@{}}Capacity\\Increase\end{tabular}& \begin{tabular}{@{}c@{}}Response Time\\Objective\end{tabular} & \begin{tabular}{@{}c@{}}Competitive\\Ratio\end{tabular}   \\
    \hline
    \hline \\[-1em]
    LP Batch & $4$ & maximum & $2$ \cite{jahanjou+rajaraman+stalfa.flowswitch.20} \\
    \\[-1em] \hline \\[-1em]
    \begin{tabular}{@{}l@{}}Equal\\Priority\end{tabular} & $1 + \varepsilon$ & maximum & $1/\varepsilon$ * \\
    \\[-1em] \hline \\[-1em]
    \begin{tabular}{@{}l@{}}Shortest\\Job First\end{tabular} & $2 + \varepsilon$ & average & $2(2+\varepsilon) / \varepsilon$ \cite{dinitz+moseley.reconfigurable.20} \\
    \\[-1em] \hline \\[-1em]
\end{tabular}
\end{center}
\begin{flushright}
\footnotesize
* indicates results in this paper
\end{flushright}
\medskip
}

Our upper bound results for maximum response time and previous bounds for maximum and average response time are summarized in Table~\ref{tab:bounds}.  We note that the LP Batch algorithm does not split jobs, and so assumes that all requests have demands less than the capacities of the nodes on which they arrive. This presents a limitation for LP Batch in the most general setting. Proportional Allocation and Shortest Job First make no such assumption about the size of the requests. However, in the case where jobs are unsplittable, Proportional Allocation and Shortest Job First are not well defined. Both split jobs even when all jobs can be routed in a single round on both of their nodes. 

\paragraph{Simultaneous approximation of response time metrics.}
The upper bound results suggest one should use the Shortest Job First algorithm to optimize average response time, and one should use Proportional Allocation to optimize for maximum response time.  It is natural to ask whether there is an algorithm that is simultaneously competitive for both average and maximum response times.  We find that 
Proportional Allocation has poor competitive ratio for average response time, even under arbitrary resource augmentation (Lemma \ref{thm:propalloc_avg}), while Shortest Job First has poor competitive ratio for maximum response time under resource augmentation at most two (Lemma \ref{thm:shortestfirst_max}). 
However, the upper bounds for these two algorithms yield a new hybrid algorithm which provides guarantees for both average and maximum response time. In each round, route jobs according to Equal Priority with $\varepsilon_1$ and also according to Shortest Job First with $\varepsilon_2$. This specifies an algorithm that is $1/\varepsilon_1$ competitive for maximum response time and $2(2+\varepsilon_2)/\varepsilon_2$ competitive for all $\ell_p$ norms ($p \in [1,\infty)$), where the algorithm increases the capacity on each port by a factor of $3 + \varepsilon_1 + \varepsilon_2$.

For unsplitabble jobs in the unit case, one can establish a better bound.  We show that for any positive integer $k$, a simple FIFO algorithm with resource augmentation of $2 + k$
achieves a competitive ratio of $\max\{1, 2/k\}$ for maximum response time, while previous work of~\cite{dinitz+moseley.reconfigurable.20} established a bound of $2(2 +k)/k$ for average response time.
\junk{
which increases all node capacities by a factor of $2 + \varepsilon$, achieves the following competitive ratios:  $(2(2+\varepsilon)/\varepsilon)$ for average response time and for all $\ell_p$ norms with $p \in (1,\infty)$, and $2/\varepsilon$ for maximum response time (Theorem \ref{thm:FIFO_upperbounds}).}

\junk{
\paragraph{General Lower Bounds.}
In our second set of main results, we present general lower bounds for Endpoing Capacitated Flow Scheduling. 

We show that, without resource augmentation, no algorithm performs well on the metrics we consider. The first such result is an adaptation of the construction given in \cite{irani+leung.conflicts.96} for maximum response time. 

\begin{theorem}
Consider the case of unit demands and unit capacities. For any deterministic online splitting ECFS algorithm with max response time $R_{Alg}$, there exists a nonsplitting schedule with max response time $R^*$ such that $R_{Alg}/R^* = \Omega(n)$, where $n$ is the number of ports. \comment{David is working on this}
\end{theorem}

This result is complemented by a lower bound on the competitiveness of any algorithm for any $\ell_p$ norm with $p \in [1,\infty)$ (Theorem \ref{thm:pnorm_lowerbound}). This proof uses the lower bound construction presented in \cite{jahanjou+rajaraman+stalfa.flowswitch.20, dinitz+moseley.reconfigurable.20}.

\paragraph{Unit Demands.}\comment{do we need to assume unit capacity as well?}
As we noted above, Equal Priority and Shortest Job First are both ill defined in the setting where jobs are unsplittable. LP Batch is well defined in this space, and provides good performance in the most general setting, but its resource augmentation is relatively high in the setting of unit demands and capacities. These considerations motivate our analysis of the unit demand setting, in which we prove performance guarantees for two algorithms. The first is a combinatorial batching algorithm which is $(1+\varepsilon)$-competitive for maximum response time with node capacities increased by a factor of $1 + 2/\varepsilon$.   
}
\paragraph{Techniques.}
The algorithms studied in this paper explore simple techniques for scheduling flows: spreading the flows proportional to load (Proportional Allocation), assigning priorities to jobs based on size or release time (Shortest Job First and FIFO), and scheduling jobs in batches (Batch Decomposition).  Our results indicate that these different approaches yield different tradeoffs for response time.   

All of our upper bounds for maximum response time use the same lower bound on the optimum, which we refer to as the \textit{interval lower bound}. Informally, we consider every interval of input rounds over the course of a given \ecfs\ instance.  Whenever the total load placed on any node during the interval exceeds the length of the interval, we obtain a non-trivial lower bound on the maximum response time, given by the difference between the load and the total capacity of the node over the length of the interval. \junk{For example, if ten unit jobs arrive on a unit capacity port in a single round, then the best response time you can achieve is by executing one of these edges in the round it arrives, and the nine others in the immediate following rounds. Similarly, if a hundred jobs arrive on one node over thirty rounds, then the best response time one can hope for is seventy.}
The interval lower bound is local and applies equally well for the much simpler setting where jobs require the resources of a single node rather than a pair of nodes. In fact, as we prove in Lemma~\ref{thm:interval_gap}, there are instances for which the gap between the interval lower bound and the optimal max response time is quite large. Given the strong lower bound presented in Theorem~\ref{thm:lower}, it is therefore surprising that algorithms with a small amount of resource augmentation are competitive with this weak lower bound, as our upper bounds show.

\junk{
\begin{theorem}
batch algorithm is 2-competitive with max response time with 2 resource augmentation. \comment{Dang is working on this, for non-splittable case}
\end{theorem}

\begin{theorem}
There exists an instance with interval lower bound $L$ for which the optimal maximum response time is $\Omega(\sqrt{n} \cdot L)$.
\end{theorem}

\begin{theorem}
There exists an offline $o(1)$-approximation algorithm to optimize makespan. \comment{Rajmohan is working on this}
\end{theorem}
}

\setcounter{thm}{0}

\subsection{Related work.}

A closely related model studies the scheduling of flows on optical WAN, modelled as reconfigurable networks. Jia et al.\ \cite{jia+etal.opticalwan.17} initiated the study of scheduling on optical WAN under the makespan and sum of completion times objective. Using the same model,  Dinitz and Moseley \cite{dinitz+moseley.reconfigurable.20} provide an algorithm that is $O(1/\varepsilon^2)$-competitive for the weighted sum of response time objective, assuming that the speed of each machine is $2 + \varepsilon$ times that in an optimal solution. Recent work by Kulkarni et al.\ \cite{kulkarni+shmid+schimidt.twotiered.21} extend these same bounds to a two-tiered reconfigurable network in which one tier of the network is fixed and the other tier is reconfigurable. The notion of resource augmentation in these papers is subtly different from ours. These models use \textit{speed} augmentation, which allows for faster completion of a job, while we use \textit{capacity} augmentation which allows more unit flows to be scheduled in each round. This entails that their upper bounds immediately extend to our model, though some of ours are not valid in theirs.  

An identical model to \ecfs\ with unsplittable flows is studied in Jahanjou et al.\ \cite{jahanjou+rajaraman+stalfa.flowswitch.20}, which presents algorithms for the sum of response times and maximum response time objectives. In the setting where the jobs have unit demand and form a bipartite graph on the nodes, the paper provides an offline $O(\log n)/c$-approximation  algorithm for the sum of response times objective with a $1+c$ factor increase on the capacity of each node. For the maximum response time objective, the paper provides an optimal schedule in the offline setting with $2d_{max}-1$ capacity augmentation, where $d_{max}$ is the maximum demand of all requests. The paper also presents a 2-competitive online algorithm for max response time with capacity increase of 4, which uses the offline algorithm as a subroutine. Finally, Jahanjou et al.\ prove the NP-hardness of constructing an \ecfs\ schedule which minimizes maximum response time.

The more general setting of conflict scheduling is studied in 
\cite{irani+leung.conflicts.03, irani+leung.probabalistic.97, even+etal.conflict.09}. In this model, a node represents a class of jobs, and two nodes are adjacent if the corresponding two classes of jobs requre the same resource. In each round, only jobs belonging to an independent set of nodes can be executed. \ecfs\ can be reduced to conflict scheduling by considering the line graph of the requests. For the maximum response time objective, Irani and Leung \cite{irani+leung.conflicts.03} provide an $O(n^3 \opt^2)$-competitve algorithm for interval graphs and an $O(n^2 \opt^2)$-competitive algorithm for bipartite graphs, where $n$ is the number of nodes in the conflict graph. We note that these results do not imply upper bounds in our model because \ecfs\ instances may transform to conflict graphs that are neither bipartite nor interval. The paper also provides a strong $\Omega(n)$ lower bound on the competitive ratio of any algorithm for max response time. To overcome this lower bound, Irani and Leung \cite{irani+leung.probabalistic.97} consider the probabilistic version of model, where a job arrives on node $i$ with probability $p_i$ in each round.

%packet scheduling  
In a related setting, Guez et al.\ \cite{Guez} studies the problem of maximizing the throughput of packets on input-queued switches. Packets are queued at input ports, and are forwarded to output ports through a cross switch fabric. The main difference with \ecfs\ is that their model allows a packet to be dropped, while ours requires that all requests be served. There is extensive literature on scheduling matchings over high-speed crossbar switches and these studies largely adopt a queuing theoretic framework (e.g., see~\cite{Giaccone26,Gong28,Shah50}). Approximation algorithms for average completion time of co-flows, a collection of flows that share a common performance goal, on a non-blocking switch are given in \cite{DBLP:journals/ton/ShafieeG18} and \cite{khuller38}.

% dang
\section{Model}
An instance of \ecfs\ is composed of the following elements. We are given a set $V$ of $n$ \textit{nodes}. Each node $i \in V$ has a corresponding capacity $c_i$. There is also a set $J$ of $m$ \textit{requests} or \textit{jobs}. Each request $j$ has an associated set of nodes $\{i,i'\}$,  demand $d_j \in \mathbb{R}^+$, and release $r_j \in \mathbb{N}^+$. For a given node $i$, we say that $i \in j$ if $i$ is contained in the set of nodes associated with $j$. In each round, we may execute any released requests subject to the constraint that for any port $i$, the total demand of all executed requests $j$ adjacent to $i$ is no more than $c_i$. More formally, if $x_{j,t}$ is the fraction of request $j$ executed in round $t$, then, $x_{j,t}>0$ implies that $r_j \le t$ and, for every node $i$, we have that $\sum_{j \ni i} d_jx_{j,t} \le c_i$. We also require that, over the course of the entire schedule, every job is completely executed, i.e. for all jobs $j$, we have that $\sum_t x_{j,t} \ge 1$. In any round $t$ in the schedule, if there is a request $j$ with $r_j \le t$ and $j$ has not been completely scheduled, then we say that $j$ is \textit{pending}. 

For a given instance, the goal is to construct a schedule that has good performance relative to some objective. Here, we consider two objectives: maximum response time and average response time. The response time of a job $j$ is the difference between $j$'s release time $r_j$ and its completion time $C_j$, which is the first round in which the job is completed, i.e. the first round after $r_j$ in which $j$ is no longer pending. Note that if a job is executed in the same round that it arrived, its response time is 1. For a given schedule, the maximum response time is the maximum amount of rounds between the release of a job and its completion: $\max_j \{C_j - r_j\}$. The average response time of a schedule is the sum of all response times over the number of jobs: $\left( \sum_j C_j - r_j \right)/m$. 

We distinguish between splitting and nonsplitting schedules. A \textit{splitting} schedule is permitted to schedule fractions of a job over multiple rounds, i.e. for any job $j$, the fraction of $j$ routed in round $t$ may be between 0 and 1. On the other hand, a \textit{nonsplitting} schedule must excecute the job entirely in a round or not at all, i.e. for any job $j$, the fraction of $j$ routed in round $t$ is either 0 or 1. 

We consider the problem in the online setting. In this setting, at the start of each round $t$, an algorithm receives the set of jobs $j$ such that $r_j = t$. The algorithm must then irrevocably schedule any of the pending jobs in that round before moving to the next. We analyze the performance of such algorithms using competitive analysis, where the objective value of the algorithm is compared against the offline optimum~\cite{sleator+t:online,borodin+e:online}.  Formally, if \alg\ is the algorithm's objective value and \opt\ is the offline optimum, then the \textit{competitive ratio} is given by the ratio of \alg\ to \opt.

We also allow our algorithms to use some amount of \textit{resource augmentation}. If an algorithm uses resource augmetation $\gamma$, then in any round the total load that can be scheduled on node $i$ is $\gamma \cdot c_i$. Our competitive analysis then compares the objective value achieved by the algorithm against the offline optimum \textit{without resource augmentation.}

% david (rajmohan read on sunday)
\section{Lower Bound without Resource Augmentation}
\label{sec:lower}

In this section we show that there is no algorithm for \ecfs\ that achieves good competitive ratio for maximum response time without resource augmentation, even for the unit case when all demands and capacities are unit. This result (Theorem~\ref{thm:lower} restated below) motivates our use of resource augmentation in the remainder of the paper.  

\ifNoFormat
\lowerbound*
\fi 

\ifAPOCS
\begin{thm}
For any deterministic online \ecfs\ algorithm Alg that is allowed to split jobs, there exists an instance with unit demands and capacities such that the maximum response time of Alg is $\Omega(n)$ times that achieved by an optimal schedule that does not split jobs, where $n$ is the number of nodes.
\end{thm}
\fi

The construction used to prove Theorem~\ref{thm:lower} closely follows the one used in the elegant lower bound proof for the conflict scheduling problem in~\cite{irani+leung.conflicts.03}.  We note, however, that while
the setting of~\cite{irani+leung.conflicts.03} is more general, their result does not immediately extend to \ecfs\ because they do not allow for splittable jobs. 

We consider a chain with $n = 4K + 5$ nodes for any $K \ge 1$. Starting from the left, label every edge in the path from $-2K-2$ to $2K+1$. For each edge with even label $\ell \ge 2K$, the left node of the edge is labelled $\ell/2$. Figure \ref{fig:LB_chain} shows the construction. Fix an arbitrary integer $C \ge 1$. 
We define the adversary Adv using Input Sequences~\ref{alg:adversarial_subroutine} and~\ref{alg:adversarial_input}, the former being a subroutine used in the latter, which define the adversarial input.   Input Sequence~\ref{alg:adversarial_subroutine} has two parameters $k \le K$ and $c \le C$. Figure~\ref{fig:LB_inputex} provides examples of input patterns generated by the subroutine. 

In any round of a schedule, we say that an edge has \textit{load} $d$ if, at the start of the round, the sum of the fractional parts of pending requests on the edge is $d$. Similarly, we say that a node $i$ has load $d$ in a round if the sum of the loads of both edges adjacent to $i$ in the round is $d$. 

\begin{figure}
    \centering
    \ifAPOCS \includegraphics[]{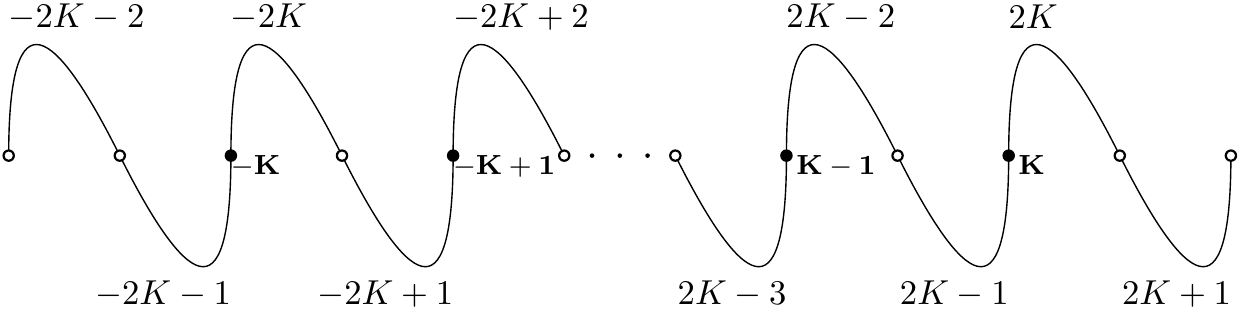} \fi
    \ifNoFormat \includegraphics[width = \columnwidth]{Figures/LB_chain.pdf} \fi 
    \caption{A chain of $4K+5$ nodes. Each edge is labelled from $-2K-2$ to $2K+1$. Each black node is labelled from $-K$ to $K$.}
    \label{fig:LB_chain}
\end{figure}

{\SetAlgorithmName{Input Sequence}{input}{List of Input Sequences}

\begin{algorithm2e}
\KwData{parameters $k \le K$ and $c \le C$}
$t \leftarrow 1$\;
\Repeat{the load on $i$ in Alg is at least $(Ck + c + 1)/2$ \label{line:terminate}}{
    $\ell \leftarrow \lceil t/C \rceil - 1$\;
    the following requests arrive in a single round: one on each odd numbered edge adjacent to nodes $-k, -k+1, \ldots, -(\ell + 1)$,
    one on each edge adjacent to nodes $-\ell, -\ell + 1, \ldots, \ell$, one on each even numbered edge adjacent to nodes $\ell + 1, \ell+2 \ldots, k$\; \junk{
    a request arrives on each odd numbered edge adjacent to nodes $-k, -k+1, \ldots, -(\ell + 1)$\; 
    a request arrives on each edge adjacent to nodes $-\ell, -\ell + 1, \ldots, \ell$\;
    a request arrives on each even numbered edge adjacent to nodes $\ell + 1, \ell+2 \ldots, k$\;
    \textbf{if} \textit{at the end of round $T + t$ there is some unlabelled node $i$ with excess load $(Ck+c+1)/2$ in Alg} \textbf{then return} $i$\;}
    $i \leftarrow$ the unlabelled node with maximum excess load in Alg after scheduling the current round\; 
    $t \leftarrow t + 1$
}
\Return i
\caption{Adversarial Subroutine}
\label{alg:adversarial_subroutine}
\end{algorithm2e}
}

The following two lemmas give the basic strategy the adversary employs. By introducing jobs as in the subroutine, the adversary guarantees that Alg accrues a certain amount of extra load on some node while Adv is able to keep that node empty. By carefully chaining successive runs of this subroutine together, one after the other, the adversary is able to build up a significant amount of excess load in Alg, all while keeping its own schedule Adv mostly empty. Note that Lemma~\ref{lem:lowerbound_increase} also implies that Input Sequence~\ref{alg:adversarial_subroutine} always terminates since the condition at line~\ref{line:terminate} is always eventually met.
 
\begin{lemma}
Suppose that Alg has $(Ck+c)/2$ load remaining on each odd edge in round $T$. Then after $Ck + c + 1$ iterations of Input Sequence~\ref{alg:adversarial_subroutine}, there will be some unlabelled node $i$ such that Alg has $(Ck+c+1)/2$ excess load on $i$.
\label{lem:lowerbound_increase}
\end{lemma}

\begin{figure*}
    \centering
    \includegraphics[width=\textwidth]{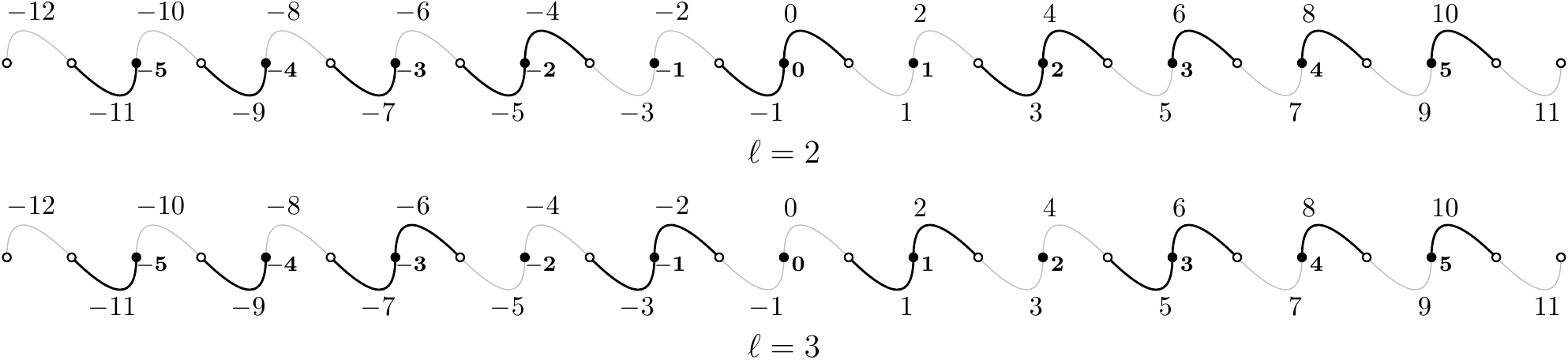}
    \caption{Two request arrival patterns for $k=5$. Top shows arrival pattern for $\ell = 2$, bottom shows patter for $\ell = 3$. Requests arrive on the bold edges.}
    \label{fig:LB_inputex}
\end{figure*}

\begin{proof}
Suppose that Alg has pending load $(Ck + c)/2$ on all odd labelled edges in round $T$. In each round, Input Sequence~\ref{alg:adversarial_subroutine} introduces a new job on every unnumbered node (except the first and last) in the path from edge $2k-2$ to $2k+1$. Note that no job ever arrives on edge $-2k - 2$. So, in order to avoid an excess load of $(Ck + c + 1)/2$ on any node, the fractional components that Alg schedules on each odd edge in the path must sum to at least $1/2$ in each round. Since a request arrives on edge $2k$ in each round, after $Ck+c+1$ rounds there is at least $(Ck+c+1)/2$ load on edge $2k$. 
\end{proof}

\junk{
We now begin to define the schedule $\sigma^*$ against which we will evaluate Alg. When scheduling the input generated by Input Sequence~\ref{alg:adversarial_subroutine},  $\sigma^*$ takes the following actions. Let $i$ be the node returned by Input Sequence~\ref{alg:adversarial_subroutine}. For a given iteration of the for-loop with fixed value of $\ell$, route all jobs that just arrived on the odd edges adjacent to nodes $-k, -k+1, \ldots, -(\ell+1)$ and all jobs that just arrived on the even edges adjacent to nodes $\ell + 1, \ell + 2, \ldots, k$. If no request on $i$ has been routed, schedule all remaining even or odd edge requests depending on if the oldest pending job on $i$ is even or odd. Otherwise, schedule all remaining even or odd requests depending on which has the oldest pending request.}

\begin{lemma}[immediate corollary to Lemma 3.3 in \cite{irani+leung.conflicts.03}]
There exists a schedule $\sigma$ over just the jobs introduced by  Input Sequence~\ref{alg:adversarial_subroutine} such that $\sigma$ has maximum response time $C$. Further, if $T$ is the last round that Input Sequence~\ref{alg:adversarial_subroutine} produces input, then in round $T+1$ the returned node $i$ has no load in $\sigma$ and, in all rounds after $T$, $\sigma$ schedules jobs entirely on even edges or entirely on odd edges.
\label{lem:sigma}
\end{lemma}

We now define the adversarial input as Input Sequence~\ref{alg:adversarial_input}, which is parametrized by integers $K$ and $C$, and uses Input Sequence~\ref{alg:adversarial_subroutine} as a subroutine.  The sequence has $KC$ iterations.  In each iteration, the sequence creates an instance of Input Sequence~\ref{alg:adversarial_subroutine} in order to build up load on one node $i$ in Alg. It then allows Adv to empty any load it has remaining while maintaining the load on $i$ in Alg. Finally, the sequence spreads the load in Alg to all the odd edges to prepare for the next iteration.

{\SetAlgorithmName{Input Sequence}{input}{List of Input Sequences}
\begin{algorithm2e}
\KwData{parameters $K$, $C$}
\For{$k= 0$ to $K$}{
    \For{$c = 0$ to $C$ \label{line:cfor}}{
        $i \leftarrow$ result of running Input Sequence~\ref{alg:adversarial_subroutine} on $k, c$\;
        \tcp{Adv schedules the jobs from Input Sequence~\ref{alg:adversarial_subroutine} according to Lemma~\ref{lem:sigma}}
        \For{$C$ iterations \label{line:emptysigmastar}}{
            \uIf{Adv schedules odd edges in the next round}{
                introduce a job on the odd edge adjacent to $i$
            }
            \Else{introduce a job on the even edge adjacent to $i$}
            %\textbf{if} Adv schedules odd edges in the next round \textbf{then} introduce a job on the odd edge adjacent to $i$ \;
            %\textbf{else} introduce a job on the even edge adjacent to $i$\;
            \tcp{Adv executes the new job immediately}
        }
        \For{$2K+2$ iterations \label{line:spreadexcess}}{
            introduce one request on every even numbered edge for $Cn$ rounds \label{line:evenspread}\;
            introduce one request on every odd numbered edge for $Cn$ rounds \label{line:oddspread}\;
            \tcp{Adv executes each new job as it arrives}
        }
    }
}
\caption{Adversarial Input}
\label{alg:adversarial_input}
\end{algorithm2e}
}

\begin{lemma}
On the input generated by Input Sequence~\ref{alg:adversarial_input}, Alg has maximum response time $C \cdot (n-5)/4$.
\label{lem:algmaxresponse}
\end{lemma}

\begin{proof}
We prove the lemma by showing that, for any given values of $k \le K$ and $c \le C$, at the start of iteration $k \cdot c$ of the line~\ref{line:cfor} for-loop, Alg has $(Ck + c)/2$ pending load on all odd edges. The claim is proved by induction on $k$ and $c$. The base case where $k = c = 0$ is trivial. 

Suppose that at the start of iteration $k \cdot c$, Alg has load $(Ck+c)/2$ pending on all odd edges. After running Input Sequence~\ref{alg:adversarial_subroutine} on input $k$, $c$, Alg has load $(Ck + c + 1)/2$ on some node $i$, by Lemma~\ref{lem:lowerbound_increase}. So, we must establish that the remainder of Input Sequence~\ref{alg:adversarial_input} spreads that load to all odd edges. We note that the line~\ref{line:spreadexcess} for-loop does not affect the load on $i$ in Alg: since one edge arrives on $i$ in every round, at the end of this for-loop there remains $(Ck + c + 1)/2$ load on $i$ as long as Alg keeps $i$ executing jobs at capacity in each round. 

We show, by induction on $t$, that after $t$ iterations of the line~\ref{line:spreadexcess} for-loop, either Alg has maximum response time of $\Omega(Cn)$ or there are at least $t$ consecutive odd edges with load $(Ck + c + 1)/2$ in Alg. The claim is trivial for $t = 0$, so assume the claim holds after $t \ge 0$ iterations. In this case, when line~\ref{line:evenspread} is executed, Alg can avoid having maximum response time $Cn$ only by scheduling all jobs on the previously loaded odd edges. This places load $(Ck + c + 1)/2$ load on at least $t + 1$ consecutive even edges in Alg. Similar reasoning shows that after executing line~\ref{line:oddspread}, at least $t+1$ odd edges have load $(Ck + c + 1)/2$ load in Alg after executing line~\ref{line:oddspread}. 

The lemma then follows from the fact that there are $2K + 2$ odd edges and setting $K = (n-5)/4$.
\end{proof}

\begin{lemma}[immediate corollary to Theorem 3.1 in \cite{irani+leung.conflicts.03}]
Adv has maximum response time $C$.
\label{lem:sigmastar}
\end{lemma}

The proof of Lemma~\ref{lem:sigmastar} (following Theorem 3.1 in \cite{irani+leung.conflicts.03}) uses properties of Adv that we have excluded for clarity. Essentially, each time Input Sequence~\ref{alg:adversarial_subroutine} exits, Adv ensures that there is no load on Alg's overloaded node $i$. The for-loop in line~\ref{line:emptysigmastar} of Input Sequence~\ref{alg:adversarial_input} then allows Adv to schedule all remaining jobs while maintaining Alg's excess load on $i$. Adv can then easily handle the input introduced by the for-loop in line 8 of Input Sequence~\ref{alg:adversarial_input} because all these requests can be handled in the round in which they arrive.

Theorem~\ref{thm:lower} is entailed by Lemmas~\ref{lem:algmaxresponse} and \ref{lem:sigmastar}.

% david
\section{Upper Bounds for Maximum Response Time}
\label{sec:upper}

In this section, we present resource augmented algorithms for \ecfs\ that are competitive for the maximum response time objective. Our use of resource augmentation is justified by the lower bound presented in Section~\ref{sec:lower}. Since \ecfs\ schedules are constrained by node capacities, we take resource augmentation to indicate increased node capacity. 

Our upper bounds for maximum response time use the \textit{interval lower bound}. Informally, we consider every interval of input rounds over the course of a given \ecfs\ instance. If the load placed on any node during the interval is greater than the length of the interval, this provides a lower bound on the maximum response time for the entire instance. For example, if ten unit jobs arrive on a unit capacity port in a single round, then the best response time you can achieve is by executing one of these edges in the round it arrives, and the nine others in the immediate following rounds. Similarly, if a hundred jobs arrive on one node over thirty rounds, then the best response time one can hope for is seventy.

Formally, the interval lower bound $L$ for a given \ecfs\ instance is defined as follows. 
\begin{equation}
L =  \max_{i,t_1 \le t_2} \left\{ \ \frac{1}{c_i} \cdot \sum_{t=t_1}^{t_2} \sum_{j \ni i: r_j = t} d_j - (t_2 - t_1 + 1)   \right\}  + 1
\label{eq:interval_lower_bound}
\end{equation} 
In general, we are interested only in \ecfs\ instances for which $L \ge 2$. If $L \le 1$, this implies that in every round $t$ the load that arrives on $i$ is no more than $c_i$. In this case, each job can be executed in the round in which it arrives and the problem is easily solvable. 

Our first two lemmas establish fairly simple properties of the interval lower bound. The first shows that it lower bounds the optimal maximum response time. The second shows that for any interval of rounds, the input load on any port in that interval is no more than the length of the interval plus the interval lower bound. 

\begin{lemma}
For a given instance of \ecfs\ with interval lower bound $L$, the optimal maximum response time is at least $L$.
\label{lem:interval_lowerbound}
\end{lemma}

\begin{proof}
Let $\sigma$ be an arbitrary schedule and let $L$ be the interval lower bound for the given \ecfs\ instance. Fix a node $i$ and rounds $t_1 \le t_2$. The maximum load that can be executed in $\sigma$ on node $i$ in the interval $[t_1, t_2]$ is $c_i (t_2-t_1+1)$. Therefore, there is at least $\ell = \sum_{t=t_1}^{t_2} \sum_{j \ni i: r_j = t} d_j - c_i(t_2 - t_1 +1)$ load remaining on $i$ after time $t_2$. So the last job $j$ from the interval is executed no earlier than round $t_2 +  \ell/c_i$. Since $r_j \le t_2$, this yields a maximum response time of at least $\ell/c_i + 1 \le L$.
\end{proof}

\junk{
\begin{lemma}
Suppose the interval lower bound for a given instanc is $L$. Then for any fixed node $i$ and rounds $t_1 \le t_2$, we have that $\sum_{t=t_1}^{t_2} \sum_{j \ni i: r_j = t} d_j \le c_i(t_2 - t_1 +1   + L)$. \comment{this lemma is literally just a rearranging of terms in the definition...do we want to just get rid of it? doesn't look like we reference it anywhere}
%\label{lem:interval_load}
\end{lemma}

\begin{proof}
Suppose the lemma were false for some node $i$ and pair of rounds $t_1 \le t_2$. Then, by supposition, we have that $\sum_{t=t_1}^{t_2} \sum_{j \ni i: r_j = t} d_j > c_i(t_2 - t_1 +1 )  + L$. However, this entails
\begin{align*}
    L &\ge \frac{1}{c_i} \sum_{t=t_1}^{t_2} \sum_{j \ni i: r_j = t} d_j - (t_2- t_1 + 1) &&\text{by definition of } L \\
    &> \frac{ c_i(t_2 - t_1 +1  + L)}{c_i} - (t_2 - t_1 + 1) = L
\end{align*}
which is a contradiction.
\end{proof}
}

We note that the interval lower bound is local and applies equally well for the special case where jobs require the resources of a single node rather than a pair of nodes.  A natural question arises about the strength of the interval lower bound.  Our final result regarding the interval lower bound establishes that, in some cases, there is a large gap between it and the optimal maximum response time. This gap arises from the fact that the interval lower bound is merely counting the load that arrives on a given port in an interval, and does not account for other conflicts between different requests.  

\begin{lemma}
There exists an instance with interval lower bound $L$ for which the optimal maximum response time is $\Omega(\sqrt{m} \cdot L)$, where $m$ is the number of requests.
\label{thm:interval_gap}
\end{lemma}

\begin{proof}
The gap instance is constructed over four unit capacity nodes $a,b,c,d$ using unit demand jobs. Fix an arbitrary integer $C \ge 2$. A subinstance is composed of $2C$ consecutive rounds. In the last round no jobs arrive. In all other rounds of the subinstance, a request arrives on nodes $a$ and $b$. In the first round, an additional job arrives on $b$ and $c$, and in the $C$\textsuperscript{th} round an additional job arrives on nodes $a$ and $d$. The gap instance is constructed from $C$ consecutive subinstances. Figure \ref{fig:interval_gap} shows the construction.

\begin{figure}
    \centering
    \ifAPOCS \includegraphics[]{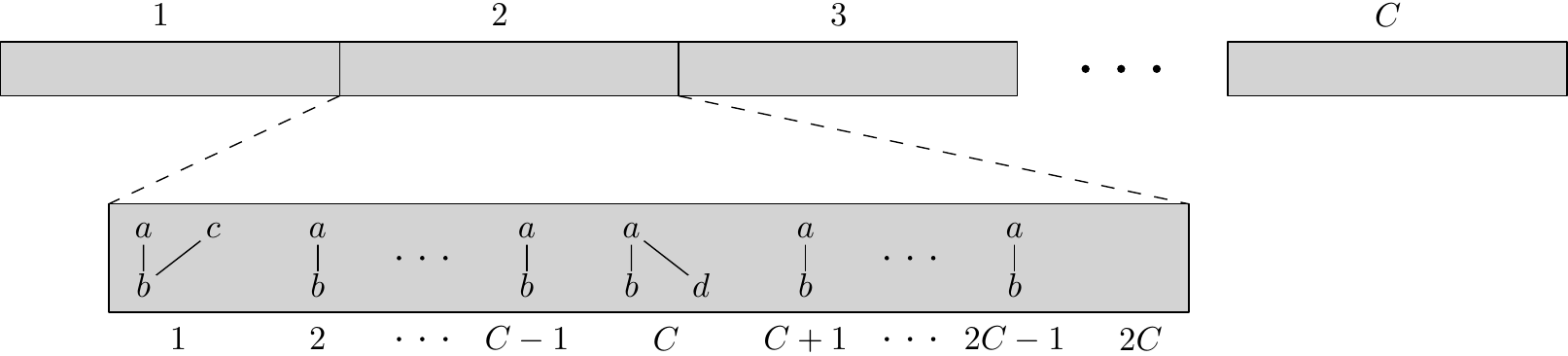} 
    \caption{Input sequence for interval lower bound gap instance. The upper sequence shows the construction of the entire instance, consisting of $C$ subinstances. The lower sequence shows the construction of a single subinstance, consisting of $2C$ rounds.}
    \fi
    \ifNoFormat \includegraphics[width=\columnwidth]{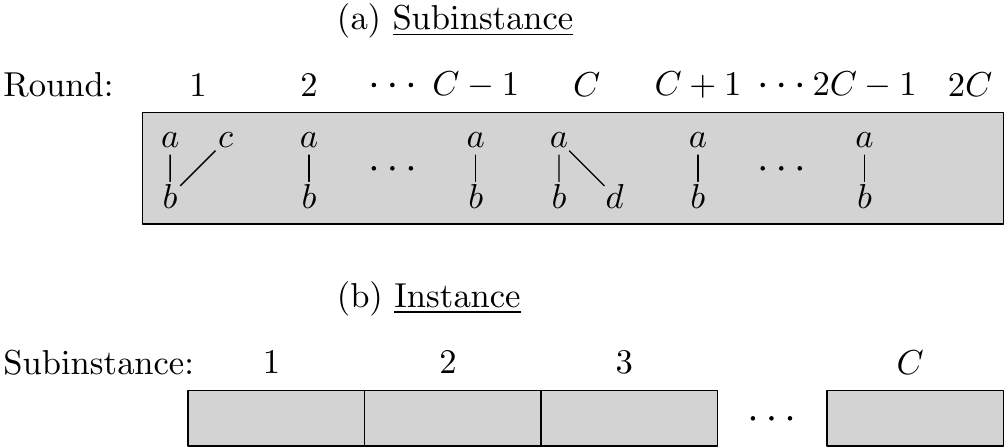} 
    \caption{Input sequence for interval lower bound gap instance. (a) shows the construction of a single subinstance, consisting of $2C$ rounds. (b) shows the construction of the entire instance as a chain of subinstances.}
    \fi
    \label{fig:interval_gap}
\end{figure}

Considering the constructed instance, we prove the theorem via the following two claims:  (1) the interval lower bound is $2$ and (2) the optimal maximum response is at least $C$. Claims (1) and (2) are sufficient to prove the theorem since the total number of requests is $2C^2+C$: $2C+1$ requests per subinstance and $C$ subinstances.

We first prove (1). Nodes $a$ and $b$ are the only nodes that ever have more than one job arrive in a single round, so we need not consider nodes $c$ and $d$ when calculating the interval lower bound. We focus our attention on node $a$ (the reasoning for $b$ is similar). Note that, between any two rounds in which more than one job is introduced on $a$, there is at least one round in which no jobs are introduced on $a$. This entails that for any interval of rounds, the number of jobs that arrive on $a$ is no more the length of the interval plus 1. This proves claim (1).

We now prove claim (2). Let $\sigma$ be a schedule with maximum response time strictly less than $C$. For a fixed subinstance, there is exactly one job $j_{a,d}$ on nodes $a$ and $d$ and exactly one job $j_{b,c}$ on nodes $b$ and $c$. In order for $\sigma$ to avoid response time $C$, it must execute job $j_{b,c}$ before the release of job $j_{a,d}$. Similarly, $\sigma$ must execute job $j_{a,d}$ before the final round of the subinstance where no jobs arrive. Note that, in any round where $j_{b,c}$ or $j_{a,d}$ is executed, $\sigma$ cannot execute any job that requires both nodes $a$ and $b$. This entails that, in each subinstance, $\sigma$ can schedule no more than $2C-2$ of such jobs. Therefore, after $C$ subinstances, there are at least $C$ jobs pending that require both nodes $a$ and $b$. The last of these jobs to be executed will have response time at least $C$, so we have derived a contradiction. This proves claim (2).
\end{proof}

\subsection{General Demands and Capacities.}

We now turn to proving upper bounds on the competitive ratio for maximum response times with resource augmentation. In this section, we present an algorithm that allows for general capacities and demands, and uses only a small amount of augmentation. The algorithm is given as Algorithm~\ref{alg:propalloc}. We define $J(i,t)$ to be the set of pending jobs on $i$ in round $t$. 

\begin{algorithm2e}
\For{each round $t = 1, 2, \ldots$}{
    \ForAll{requests $j$ pending in round $t$}{
        $f_t(j) = \min_{i \in j} \{ c_i / \sum_{j' \in J(i,t)} d_{j'} \} $\;
        execute a $(1+\varepsilon) \cdot f_t(j)$ fraction of $j$
    }
}
\caption{Proportional Allocation}
\label{alg:propalloc}
\end{algorithm2e}

\ifNoFormat
\propalloc*
\fi 

\ifAPOCS
\begin{thm}
Given an \ecfs\ instance with splittable jobs, Proportional Allocation yields a $(1/\varepsilon)$-competitive schedule for maximum response time, assuming the capacity of each node is increased by a factor of $1 + \varepsilon$, for any $\varepsilon > 0$. 
\end{thm}
\fi

\begin{proof}
We prove the theorem via two claims:
\begin{enumerate*}[label=(\alph*)]
    \item in every round there is no more than load $c_i (1+\varepsilon)$ on any node $i$ and%
    \label{max:load}
    \item for all rounds $T$, all requests released prior to round $T$ are scheduled with response time at most $\lceil L/\varepsilon \rceil$. 
    \label{max:response}
\end{enumerate*}
The proof of \ref{max:load} is straightforward. For a fixed $i$ and $t$, the total scheduled load on node $i$ in round $t$ is at most
\[ \sum_{j \in J(i,t)} \frac{c_i \cdot (1+\varepsilon)}{\sum_{j' \in J(i,t)} d_{j'}} \cdot d_j = c_i \cdot (1 + \varepsilon). \]

We prove \ref{max:response} by induction on $T$, letting $\hat{L} = \ceil{L/\varepsilon}$. The claim is trivial for $T=1$ since no requests are released prior to round 1. Consider a fixed job $j^*$ released on nodes $i_1$ and $i_2$ in round $T$. We show that the total pending load on $i_1$ in any round $t \in \{T, T+1, \ldots, T+\hat{L} -1 \}$ is at most $c_{i_1}(\hat{L} + L)$. For any $\ell \in \{0,1,\ldots, \hat{L}-1\}$, the induction hypothesis entails that each request pending on $i_1$ in round $T + \ell$ is released in some round $t \in \{ (T+\ell) - (\hat{L} -1) , (T+\ell) - (\hat{L} - 2), \ldots, T+\ell\}$. By definition of the interval lower bound, the total load placed on $i_1$ in the rounds from $(T+\ell)-(\hat{L}-1)$ to $T+\ell$ is at most $c_{i_1}$ times the length of the interval plus $L$, or $c_{i_1}(\hat{L} + L)$. Similar reasoning shows that the total load on $i_2$ in any round $t \in \{T, T+1, \ldots, T+\hat{L}-1 \}$ is at most $c_{i_2}(\hat{L} + L)$. For $t \in \{T, T+1, \ldots, T+\hat{L}-1\}$, let $i[t]$ be the node $i \in \{i_1, i_2\}$ such that $f_t(j) = c_{i}/\sum_{j \in J(i,t)} d_j$. By definition of the algorithm we have that the total fraction of $j$ scheduled in rounds $T, T+1, \ldots, T+\hat{L}-1$ is 
\begin{align*}
    \sum_{t=T}^{T + \hat{L}-1} (1 + \varepsilon) f_{t}(j) \ge \sum_{t=T}^{T + \hat{L}-1}  \frac{c_{i[t]}(1 + \varepsilon)}{c_{i[t]}(\hat{L} + L)} 
    %= \sum_{T=1}^{\hat{L}} \frac{(1+\varepsilon) \cdot c_{i_{T}}}{(\hat{L} + L) \cdot c_{i_{T}}} 
    = \frac{(\hat{L}) \cdot (1+\varepsilon) }{(\hat{L} + L) } \ge 1.
\end{align*}
Therefore, all jobs released in time $t$ are completed with response time $\lceil L/\varepsilon \rceil$. The theorem now follows from Lemma~\ref{lem:interval_lowerbound}.
\end{proof}

We make two observations regarding this result. The first is that the Proportional Allocation algorithm makes full use of the resource augmentation afforded it. In other words, a schedule produced by Proportional Allocation would not be valid under a weaker notion of resource augmentation like, e.g. that used in \cite{dinitz+moseley.reconfigurable.20}. Because that model uses a \textit{speed} notion of resource rather than \textit{capacity}, with $\alpha$ resources their algorithms are able to schedule $\alpha$ individual matchings. On the other hand, with the same resources, Proportional Allocation is able to schedule an $\alpha$-matching, which for general graphs affords our algorithm more power. 

We also observe that neither the Proportional Allocation algorithm nor proofs of its guarantees rely on the fact that each job is associated with a pair of nodes, or edge, in the underlying graph. Indeed, even in the setting where each edge is specified as an arbitrary set of nodes, or an edge in the underlying hypergraph, Proportional Allocation provides exactly the same guarantees. This surprising result extends Theorem~\ref{thm:propalloc} to more general settings, including conflict scheduling of \cite{irani+leung.conflicts.03} which can be modelled as a series of node capacitated requests in a hypergraph. 

\subsection{Unit Demands and Capacities.}

In some settings we may prefer that jobs not be split over multiple rounds if it can be avoided, or we might demand that jobs never be split. In this respect, one issue with Proportional Allocation algorithm is that, except in uninteresting cases where at most one job arrives on any node in any round, the algorithm always splits jobs. In prior work \cite{jahanjou+rajaraman+stalfa.flowswitch.20}, an LP based batching algorithm was given for \ecfs\ which performs well for the maximum response time objective and does not split jobs, although the capacity increase needed for this algorithm is quite high relative to Proportional Allocation. 

These considerations motivate the search for algorithms that avoid splitting jobs, are competitive with the optimal maximum response time, and require little resource augmentation. In this section, we present an algorithm which meets these criteria, although in order to do so we must restrict ourselves to the simplified setting where all job demands are unit and all node capacities are unit. 

Our algorithm, Batch Decomposition, is specified as Algorithm~\ref{alg:batch}. Batch Decomposition uses the notion of a 2-factor decomposition. A \textit{2-factor} of a given multigraph is a spanning subgraph that has degree at most 2. A \textit{2-factor decomposition} of a given multigraph is a set $F$ of 2-factors such that, for every edge $e$ there is some 2-factor in $F$ that contains $e$. For the remainder of this subsection, we use the terms \textit{graph} and \textit{multigraph} interchangeably. 

Batch Decomposition operates by collecting a set $S$ of pending requests and representing them as a graph. The algorithm then constructs a 2-factor decomposition on this graph and executes $k$ 2-factors in each of the subsequent rounds until the set is empty, for $k \in \{1,2\}$. The algorithm collects any jobs that arrive while the 2-factors are being scheduled into the next set $S$ to be scheduled immediately after.  

\begin{algorithm2e}
    \KwData{augmentation parameter $k \in \{1,2\}$}
    $P \leftarrow \varnothing$ \tcp{stores the set of pending jobs that are not currently being executed}
    $H \leftarrow \varnothing$ \tcp{stores the set of 2-factor subgraphs currently being executed}
    \For{rounds $t = 1, 2, \ldots$}{
        $P \leftarrow P \cup \{ j : r_j = t\}$ \\
        \If{$H = \varnothing$ \label{line:H_empty}}{
            $H \leftarrow$ 2-factor decomposition of $P$\;
            $P \leftarrow \varnothing$
        } 
        choose up to $k$ 2-factors in $H$, schedule all of their requests, and remove them from $H$ 
    }
    \caption{Batch Decomposition}
    \label{alg:batch}
\end{algorithm2e}

A useful lemma originally proved by Petersen~\cite{petersen, mulder} shows that the minimum size of a graph's 2-factor decomposition is upper bound by the maximum degree of the graph.

\begin{lemma}[Petersen's 2-Factor Theorem]
\label{lem:petersen}
For every positive integer $D$, every multigraph G with maximum degree $2D$ can be decomposed into $D$ spanning subgraphs $G_1, . . . , G_D$ with maximum degree 2. 
\end{lemma}

We now use Lemma~\ref{lem:petersen} to prove Theorem~\ref{thm:batch}, which we restate below. 
%We call such a decomposed subgraph of maximum degree $2$ given by Lemma~\ref{lem:petersen}, a \textit{2-factor}.

\ifNoFormat
\batch*
\fi 

\ifAPOCS
\begin{thm}
Given an \ecfs\ instance with unsplittable jobs, unit demands, and unit capacities, Batch Decomposition yields a $\max\{2/k,1\}$-competitive schedule for maximum response time, assuming the capacity of each node is increased by a factor of $2k$, for $k \in \{1,2\}$.  
\end{thm}
\fi 

\begin{proof}
We claim that if the number of rounds between two consecutive times the condition at line~\ref{line:H_empty} is satisfied is upper bounded by $B$, then no job has response time greater than $2B$. Consider an arbitrary job $j$ released in round $r_j$ and completed in round $C_j$. Let $T$ be the earliest round $t \ge r_j$ for which the condition at line~\ref{line:H_empty} is satisfied. By the claim and the definition of Algorithm~\ref{alg:batch}, we have that $T - r_j \le B$ and that $C_j - T \le B$. Therefore, the response time of $j$ is no more than $2B$. 

We now upper bound the number of rounds that pass between two consecutive times the condition at line~\ref{line:H_empty} is satisfied in terms of the interval lower bound $L$ (see Equation~\ref{eq:interval_lower_bound}). More formally, we define $T_{\ell}$ to be the round in which the condition is satisfied for the $\ell$\textsuperscript{th} time ($T_0 =0$). We argue by induction on $\ell$ that $T_{\ell} - T_{\ell-1} \le L/k$. The base case of $\ell = 1$ is straightforward: $T_1 -T_0 = 1 - 0 \le L/k$ for nontrivial $L \ge 2$.

For the induction step, assume that $T_{\ell}-T_{\ell-1} \le L/k$ for arbitrary $\ell \ge 1$. By definition of the algorithm, the difference $T_{\ell+1} - T_{\ell}$ is equal to $(1/k)$ times the size of the 2-factor decomposition that is constructed in round $T_{\ell}$. By definition of the interval lower bound, the maximum load that arrived on any node in the interval $\{T_{\ell} + 1, T_{\ell}+2, \ldots, T_{\ell+1}\}$ is at most $L + (T_{\ell+1} - T_{\ell}) \le L + (L/k)$ by the inductive hypothesis. This entails that the degree of any node the graph decomposed in round $T_{\ell}$ is no more than $L + (L/k)$. By Lemma~\ref{lem:petersen}, we have that the size of the 2-factor decomposition of this graph is no more than $(L + (L/k))/2$.  So, by definition of the algorithm, we have that
\[ T_{\ell + 1} - T_{\ell} \le \frac{1}{k} \cdot \frac{1}{2} \cdot \left( L + \frac{L}{k} \right) = \frac{L}{2k} + \frac{L}{2k^2} \le \frac{L}{k} \]
which completes the induction and proves the theorem. 
\end{proof}

\junk{
\bigskip
\comment{OLD PROOF}
\comment{is $m$ meant to be the number of jobs? if not should change. also, maybe say batch length? or batch time? batch size sounds like it's supposed to be the number of jobs in the batch.}
We define \textit{batch size} to be the number of consecutive rounds without entering the conditional in line 4 before entering the conditional for the $m$th time. We shall now argue that the batch size never grows bigger than $B = \floor*{L/k}$, for $k = \{1, 2\}$. We prove this by induction. The base case is trivial because the first batch size is 0. The induction hypothesis is that the $m^{th}$ batch size is at most $B$. We need to prove that the $(m+1)^{th}$ batch size is also at most $B$. We will show that the set of all requests that arrive in any interval of length at most $B$ takes at most $B$ rounds to complete. 

We will show that the batch size never grows bigger than $\floor*{L/k}$. By Lemma \ref{lem:interval_load}, the number of requests that can arrive on a given port in an interval of length $l \leq B$ is at most $B + L$. Therefore, the graph of such requests $G$ has a maximum degree of at most $\Delta = B + L$. 

First, we will consider the easier case of $k=1$: we have $B = L$ and $\Delta \leq 2L$. By Lemma~\ref{lem:petersen}, we can decompose $G$ into $L$ 2-factors. In each round, with resource augmentation $2$, the algorithm can execute $1$ 2-factor. It follows that in an interval of length at most $L$, with a resource augmentation $2k$, the algorithm can execute $L$ 2-factors. Therefore, the set of all requests that arrive in any interval of length at most $B$ takes at most $B$ rounds to complete.

For $k=2$, we have $B = \floor*{L/2}$ and $\Delta = B + L$. In each round, with resource augmentation $2k$, the algorithm can execute $2$ spanning subgraphs of maximum degree 2. Let $a \in \mathbb{N}$. We consider 4 different cases. 

Case $L = 4a$: we have $\Delta = 6a$, and therefore the number of 2-factors is $3a$. With resource augmentation 4, we can execute two 2-factors in each round, which implies that we can execute $G$ in $\ceil*{3a/2} \leq 2a = \floor*{L/2}$.

Case $L = 4a+1$: we have $\Delta = 6a+1$. By creating a fake request, we can augment $G$ to have maximum degree $\Delta+1 =6a+2 $, so that we could still apply Lemma~\ref{lem:petersen}. The number of 2-factors is $3a+1$. We can execute two 2-factors in each round; therefore we can execute $G$ in $\ceil*{(3a+1)/2} \leq 2a = \floor*{L/2}$. 

Case $L = 4a+2$: we have $\Delta = 6a+2$. The same argument of the second case holds here.

Case $L = 4a+3$: we have $\Delta = 6a+4$. Therefore the number of 2-factor subgraphs is $3a+2$. We can execute two 2-factors in each round, so we can execute $G$ in $\ceil*{(3a+2)/2} \leq 2a+1 = \floor*{L/2}$.

Since any request has to wait at most $B$ rounds after its arrival to get into the decomposition step, and since it will be completed in at most $B$ steps, we get that the maximum response time of any request is $2B = 2 \floor*{L/k} \geq 2L/k$. Moreover, by Lemma~\ref{lem:interval_lowerbound}, we know that the optimum maximum response time is at least $L$. Hence, the theorem is proven.  
}

% dang
\section{Simultaneous Approximation of Response Time Metrics}
\label{sec:simultaneous}
%\comment{So far, we have considered algorithms that either optimizes average response time or maximum response time. However, in context where the response time for both the average request and the $99^{th}$ percentile request are important, it is neccessary to consider optimizing both objectives at the same time. We will explore more of the following two algorithms. }

The algorithms presented in Section~\ref{sec:upper} seek to optimize a single objective: maximum response time. In some settings, it may be desirable to have an algorithm that optimizes several objectives simultaneously. In this section, we show that we can combine our results with prior work to achieve algorithms which are competitive for both maximum and average response time simultaneously.   As above, we study two different settings for this problem: one with general demands splittable jobs, and the other with unit demands and unsplittable jobs. 

%In Section~\ref{sec:upper}, we show that the proportional allocation algorithm achieves $(1/\varepsilon)$-competitive ratio using $(1+\varepsilon)$-resource augmentation for maximum response time, and some augmentation is essential by the lower bound in Theorem \ref{thm:lower}. 

%In Theorem $\rom{4.14}$ of \cite{dinitz+moseley.reconfigurable.20}, a different algorithm, which we call Shortest Job First, was shown to achieve $O(1/\varepsilon^2)$ using $(2+\varepsilon)$-resource augmentation for average response time. The Shortest Job First algorithm simply orders all pending requests in a non-increasing order with respect to their demand, and schedules as much as possible of the requests in that order. 

%A natural question whether there exists an algorithm that simultaneously approximates both metrics.  In this section, we explore this question. \comment{First, we explore in context of Splitting Schedules for General Demands, and then we explore in context of Nonsplitting Schedules for Unit Demands.}

\subsection{Splitting Schedules for General Demands.}
An easy corollary of Theorem~\ref{thm:propalloc} and a result of Dinitz-Moseley is an existence of an algorithm with resource augmentation $3 + \varepsilon$ that is $(2/\varepsilon)$-competitive for maximum response time and $(1/\varepsilon^2)$-competitive for average response time. The algorithm simply runs each individual algorithm in parallel and sums their resource augmentation.

\begin{corollary}[Corollary to Theorem~\ref{thm:propalloc} and Theorem IV.14 in \cite{dinitz+moseley.reconfigurable.20}]
There exists an algorithm with resource augmentation $3 + \varepsilon$ that is $O(1/\varepsilon^2)$-competitive for average response time, and is $2/\varepsilon$-competitive for maximum response time. 
\end{corollary}

\begin{proof}
In each round, use $1+\varepsilon/2$ resource to execute Proportional Allocation algorithm and use $2+\varepsilon/2$ resource to execute the Shortest Job First algorithm. 
\end{proof}

We show, however, that by itself Proportional Allocation with any amount of resource augmentation performs poorly for average response time, while Shortest Job First performs poorly for maximum response time if resource augmentation is at most 2.

\begin{lemma}
Proportional Allocation with resource augmentation $\alpha$ has competitive ratio $\Omega(m/\alpha)$ for average response time, where $m$ is the number of requests, for any $\alpha > 0$. 
\label{thm:propalloc_avg}
\end{lemma}

\begin{proof}
Consider the following family of instances. At time 1, $k$ unit requests and one request with demand $M$ all arrive on a unit capacity node $i$ (no other jobs arrive and the jobs' other endpoints are irrelevant). Consider the schedule which executes all unit requests in the first $k$ rounds, and then schedules the size $M$ request in the following $M$ rounds. In this case, the sum of response times is $1 + 2 + \cdots + k + (M+k) \le M + k^2$.

On the other hand, Proportional Allocation with $\alpha$ resource augmenation schedules a $\alpha/(M+k)$ fraction of each request in each round until all requests are completed. In this case, all requests are completed in the same round no earlier than $(M+k)/\alpha$. Therefore, the sum of response times is at least $(M+k)(k+1)/\alpha \ge Mk/\alpha$. 

Setting $M = k^2$, the competitive ratio is at least $\Omega(k/\alpha)$. Since $m = k+1$, this proves the result. 
\end{proof}

\begin{lemma}
Shortest Job First with resource augmentation 2 has competitive ratio $\Omega(m)$ for max response time, where $m$ is the number of requests.
\label{thm:shortestfirst_max}
\end{lemma}

\begin{proof}
Consider the following instance over four unit capacity nodes $a,b,c,d$ as shown Figure \ref{fig:counter_shortest}. In all rounds $t = 1, 2, \ldots, T$, if $t$ is odd then two unit jobs arrive on nodes $\{a, b\}$, and if $t$ is even then two unit jobs arrive on nodes $\{c, d\}$. Additionally, one request $j^*$ arrives on nodes $\{b, d\}$ with $d_{j^*} = 2$ in round 1. 

It is straightforward to give a schedule with maximum response time 4. First execute the job $j^*$ over rounds 1 and 2, then schedule the remaining jobs in FIFO order. All jobs on $\{a, b\}$ are independent of jobs on $\{c, d\}$, so we can execute them in parallel. So, in every round $t \ge 3$ we execute one job adjacent to all nodes and in any interval (not including round 1) of length 2 or greater the number of nodes that arrive on any port is at most the length of the interval. 
%\comment{At the begining of round $2t$, the pending requests always are: two unit jobs nodes $\{c, d\}$ that just arrived and the pending requests are two unit jobs on ports $\{a, b\}$ that arrived one round ago.}
Therefore, all jobs $j \ne j^*$ are completed in time $r_j + 4$.

On the other hand, even with a resource augmentation factor of 2, Shortest Job First never schedules the request $j^*$ until round $T+1$. With resource augmentation of 2 (or less), Shortest Job First always gives preference to the newer, unit size jobs that arrive in each round on either node $b$ or node $d$, and scheduling these jobs uses all the (augmented) capacity of the node. Therefore, Shortest Job First has competitive ratio $\Omega(T) = \Omega(m)$ for maximum response time with resource augmentation of 2 or less.
\end{proof}

\begin{figure}
    \centering
    \ifNoFormat \includegraphics[width=\columnwidth]{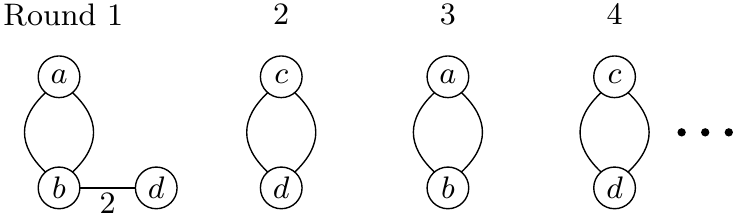} \fi
    \ifAPOCS \includegraphics[]{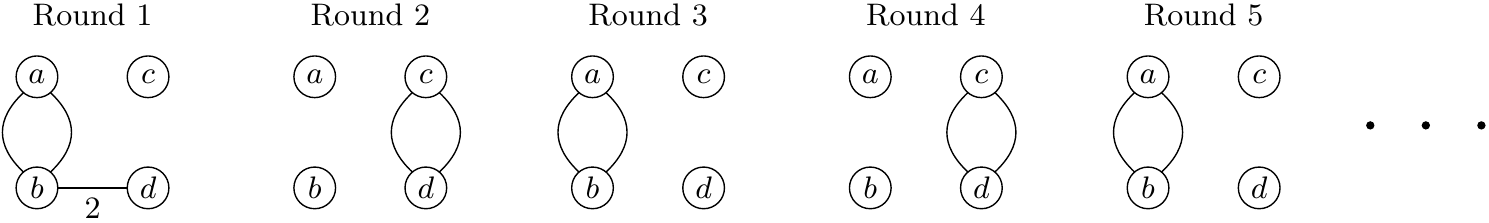} \fi
    \caption{Job arrival sequence for lower bound on max response time for Shortest Job First with 2 resource augmentation. The request that arrives on nodes $b$ and $d$ in round 1 has demand 2. All other requests have unit demand. All nodes have unit capacity. }
    \label{fig:counter_shortest}
\end{figure}

The lemma implies that, given the same resources, Shortest Job First cannot match the competitiveness of Proportional Allocation with respect to maximum response time. However, it remains open whether Shortest Job First with $2+\varepsilon$ resources achieves any well-bounded competitive ratio for maximum response time.

\subsection{Nonsplitting Schedules for Unit Demands.}
In the case of nonsplitting schedules for unit demands, we show that the simple FIFO algorithm given in Algorithm~\ref{alg:fifo}, with a resource augmentation factor of $(2+k)$ is simultaneously competitive for the average response time and maximum response time metrics, for any positive integer $k$.  We begin by analyzing FIFO under the maximum response time metric. 

\begin{algorithm2e}
\KwData{augmentation parameter $k \in \{1,2\}$}
$P \leftarrow \varnothing$ \tcp{the set of pending requests}
\For{each round $t = 1, 2, 3, \ldots$}{
    $P \leftarrow P \cup \{j : r_j = t\}$\;
    \For{each request $j = \{i_1, i_2\}$ in $P$ in release time order}{
        \If{the load on $i_1$ is less than $(2+k)c_{i_1}$ and the load on $i_2$ is less than $(2+k)c_{i_2}$}{
            schedule $j$ and $P \leftarrow P \setminus \{j\}$
        }
    }
}
\caption{FIFO Maximal Matching with a resource augmentation factor of $(2+k)$ \label{alg:fifo}}
\end{algorithm2e}

\begin{lemma}
For maximum response time metric, with a resource augmentation of a factor of $(2+k)$, FIFO Maximal Matching is $\max\{2/k,1\}$-competitive. 
\label{thm:FIFO_upperbounds}
\end{lemma}

\begin{proof}
We restrict our attention to $k \in \{1,2\}$ since the claim for $k > 2$ is subsumed by the case $k = 2$.  For ease of notation, let $\gamma = 2/k$. Since $k \in \{1,2\}$, we have that $\gamma \in \{1,2\}$. 

We prove by induction on $t$ that in any round $t + \gamma L$ there is no pending request that arrives more than $\gamma L$ rounds earlier, where $L$ is the interval lower bound give in Equation~\ref{eq:interval_lower_bound}. The base case when $t=0$ is trivially true. For the induction step, suppose that the claim holds for all $t'$ such that $0 \le t' \le t$. Suppose, for the sake of contradiction, that there exists a request $j^*$ on nodes $i$ and $i'$ such that $j^*$ is pending in round $t+\gamma L$ and $r_{j^*} \le t$. By definition of the algorithm, it must be the case that in each round $t \in \{t, t+1, \ldots, t+ \gamma L\}$, either $i$ or $i'$ is at full capacity. So we can infer that either $i$ or $i'$ is at full capacity for at least half these rounds. Without loss of generality, we assume it is $i$.  

Given the resource augmentation, the total load executed on $i$ in the interval $\{t, t+1, \ldots, t+ \gamma L\}$ is at least $(2+k) \cdot (\gamma L / 2) = \gamma L + L$ by definition of $\gamma$. Consider the set of jobs $J^*$ that comprise this load. Because FIFO Maximal Matching orders requests by their arrival time, we have that $r_j \le r_{j^*} \le t$ for all $j \in J^*$. By the inductive hypothesis, we infer that all $j \in J^*$ complete within $\gamma L$ of their release time. Since these jobs are pending at time $t$ or later, we have that $r_j \ge t - (\gamma L - 1)$ for all $j \in J^* \cup \{j^*\}$. Therefore, the number of jobs released on node $i$ in rounds $\{t-(\gamma L - 1), t - (\gamma L - 2), \ldots, t\}$ is at least $\gamma L + L +1$. However, this contradicts the definition of $L$, and so completes the induction. 

Therefore, the maximum response time of FIFO Maximal Matching is $\gamma L = (2/k) L$. Since the optimal response time is at least $L$ and no algorithm has maximum response time better than 1, FIFO Maximal Matching achieves a competitive ratio of $\max\{2/k,1\}$. 
\end{proof}

\junk{
\comment{OLD PROOF}
Suppose for contradiction that there exists a request $j$ on nodes $i$ and $i'$ that is still pending at time $t +\gamma L$ after rounds after its arrival at time $t$. Let $p$ and $q$ be the endpoints of request $j$. In each round of interval $[t, t+\gamma L]$, it must be the case that either node $p$ and node $q$ must serve to the its full capacity, because request $j$ is available in each round of the interval. Therefore, we know that there exists a node, either node $p$ or $q$, that serves to the its full capacity at least half of the time in the interval $[t, t+\gamma L]$. 

Without loss of generality, let it be $p$. In the interval $[t, t+\gamma L]$, with node capacity $(2 + k)$, node $p$ serves at least  $(\gamma L/2)(2+k) = \gamma L +  L$ by definition of $\gamma$. Because FIFO Maximal Matching algorithm orders the request by their arrival time, we have that these requests of demand at least $\gamma L +  L $ on node $p$ must have arrive before $t$. By our induction hypothesis, we also have that they must have arrived no earlier than time $t - \gamma L$. Therefore, in the interval $[t-\gamma L, t]$, the total demand on node $p$ is at least $\gamma L + L + 1$. This is a contradiction, because from the definition of the interval lower bound, in an interval of length $\gamma L$, the maximum total demand on a node $p$ is $\gamma L + L$.

Hence, each request has response time at most $\gamma L$. Since the maximum response time of OPT is at least $L$, the algorithm achieves a competitve ratio of $2/k$.
}

Moreover, we know from Theorem~\rom{4.5} of~\cite{dinitz+moseley.reconfigurable.20} that with resource augmentation $(2+k)$, FIFO is $2(2+k)/k$-competitive for average response time. Putting this result together with Lemma~\ref{thm:FIFO_upperbounds}, we obtain that with a resource augmentation of $(2+k)$ for positive integer $k$, FIFO is simultaneously $2(2+k)/k$-competitive under the average response time metric and $2/k$-competitive under the maximum response time metric.

%\input{unit}

%\input{makespan}

% raj
\section{Discussion and Open Problems}

In this paper, we have shown a linear lower bound on the competitiveness of any \ecfs\ schedule for the maximum response time metric, without resource augmentation.  We then derived improved tradeoffs between resource augmentation and competitive ratio for maximum response time, for different problem settings (general vs unit case, splittable vs unsplittable jobs) using simple algorithmic approaches.  Finally, we explored simultaneous approximation of average and response time metrics and proposed a hybrid algorithm based on Proportional Allocation and Shortest Job First for general demands and FIFO for unit case.

A number of open problems remain.  First, we would like to find the best tradeoffs possible between resource augmentation and competitive ratio for response time metrics.  Toward this end, we need to derive general lower bounds under resource augmentation.  Of special interest is whether a bounded competitive ratio can be achieved for average response time with augmentation less than two.  Second, deriving improved algorithms and bounds for simultaneous approximation of response time metrics is of practical interest.  A natural direction in this regard is to consider $\ell_p$ norms of response time metrics, which would provide a spectrum with average response time ($p = 1$) and maximum response time ($p = \infty$) at each end.  We have been able to show that the results of~\cite{dinitz+moseley.reconfigurable.20} for Shortest Job First and FIFO qualitatively extend to $\ell_p$ norms, for finite $p$, implying that the simultaneous approximations  established in Section~\ref{sec:simultaneous} generalize to $\ell_p$ norms of response times.

\ifNoFormat
\balance
\fi 

\section*{Acknowledgments}
We thank Michael Dinitz for helpful discussions on the different models and problem formulations. This work was partially supported by NSF grant CCF-1909363.

\newpage
\bibliographystyle{plain}
\bibliography{refs}

\begin{thebibliography}{10}

\bibitem{borodin+e:online}
Allan Borodin and Ran El{-}Yaniv.
\newblock {\em Online computation and competitive analysis}.
\newblock Cambridge University Press, 1998.

\bibitem{dinitz+moseley.reconfigurable.20}
Michael Dinitz and Benjamin Moseley.
\newblock Scheduling for weighted flow and completion times in reconfigurable
  networks.
\newblock In {\em IEEE INFOCOM 2020 - IEEE Conference on Computer
  Communications}, pages 1043--1052, 2020.

\bibitem{even+etal.conflict.09}
G.~Even, M.~Halld{\'o}rsson, L.~Kaplan, and D.~Ron.
\newblock Scheduling with conflicts: online and offline algorithms.
\newblock {\em Journal of Scheduling}, 12:199--224, 2009.

\bibitem{Giaccone26}
P.~Giaccone, B.~Prabhakar, and D.~Shah.
\newblock Randomized scheduling algorithms for high-aggregate bandwidth
  switches.
\newblock {\em IEEE Journal on Selected Areas in Communications},
  21(4):546--559, 2003.

\bibitem{Gong28}
Long Gong, Paul Tune, Liang Liu, Sen Yang, and Jun~(Jim) Xu.
\newblock Queue-proportional sampling: A better approach to crossbar scheduling
  for input-queued switches.
\newblock {\em Proc. ACM Meas. Anal. Comput. Syst.}, 1(1), June 2017.

\bibitem{Guez}
Dan Guez, Alex Kesselman, and Adi Ros\'{e}n.
\newblock Packet-mode policies for input-queued switches.
\newblock In {\em Proceedings of the Sixteenth Annual ACM Symposium on
  Parallelism in Algorithms and Architectures}, SPAA '04, page 93–102, New
  York, NY, USA, 2004. Association for Computing Machinery.

\bibitem{irani+leung.conflicts.03}
S.~Irani and V.~Leung.
\newblock Scheduling with conflicts on bipartite and interval graphs.
\newblock {\em Journal of Scheduling}, 6:287--307, 2003.

\bibitem{irani+leung.probabalistic.97}
Sandy Irani and Vitus Leung.
\newblock Probabilistic analysis for scheduling with conflicts.
\newblock In {\em Proceedings of the Eighth Annual ACM-SIAM Symposium on
  Discrete Algorithms}, SODA '97, page 286–295, USA, 1997. Society for
  Industrial and Applied Mathematics.

\bibitem{jahanjou+rajaraman+stalfa.flowswitch.20}
Hamidreza Jahanjou, Rajmohan Rajaraman, and David Stalfa.
\newblock Scheduling flows on a switch to optimize response times.
\newblock In Christian Scheideler and Michael Spear, editors, {\em {SPAA} '20:
  32nd {ACM} Symposium on Parallelism in Algorithms and Architectures, Virtual
  Event, USA, July 15-17, 2020}, pages 305--315. {ACM}, 2020.

\bibitem{jia+etal.opticalwan.17}
Su~Jia, Xin Jin, Golnaz Ghasemiesfeh, Jiaxin Ding, and Jie Gao.
\newblock Competitive analysis for online scheduling in software-defined
  optical wan.
\newblock In {\em IEEE INFOCOM 2017 - IEEE Conference on Computer
  Communications}, pages 1--9, 2017.

\bibitem{khuller38}
Samir Khuller and Manish Purohit.
\newblock Brief announcement: Improved approximation algorithms for scheduling
  co-flows.
\newblock In {\em Proceedings of the 28th ACM Symposium on Parallelism in
  Algorithms and Architectures}, SPAA '16, page 239–240, New York, NY, USA,
  2016. Association for Computing Machinery.

\bibitem{kulkarni+shmid+schimidt.twotiered.21}
Janardhan Kulkarni, Stefan Schmid, and Pawe\l{} Schmidt.
\newblock Scheduling opportunistic links in two-tiered reconfigurable
  datacenters.
\newblock In {\em Proceedings of the 33rd ACM Symposium on Parallelism in
  Algorithms and Architectures}, SPAA '21, page 318–327, New York, NY, USA,
  2021. Association for Computing Machinery.

\bibitem{mulder}
Henry~Martyn Mulder.
\newblock Julius {Petersen's} theory of regular graphs.
\newblock {\em Discrete Mathematics}, 100:157--175, 05 1992.

\bibitem{petersen}
Julius Petersen.
\newblock Die theorie der regul{\"a}ren graphs.
\newblock {\em Acta Mathematica}, 15:193--220, 1891.

\bibitem{DBLP:journals/ton/ShafieeG18}
Mehrnoosh Shafiee and Javad Ghaderi.
\newblock An improved bound for minimizing the total weighted completion time
  of coflows in datacenters.
\newblock {\em {IEEE/ACM} Trans. Netw.}, 26(4):1674--1687, 2018.

\bibitem{Shah50}
Devavrat Shah and Jinwoo Shin.
\newblock Randomized scheduling algorithm for queueing networks.
\newblock {\em The Annals of Applied Probability}, 22(1):128--171, 2012.

\bibitem{sleator+t:online}
Daniel~Dominic Sleator and Robert~Endre Tarjan.
\newblock Amortized efficiency of list update and paging rules.
\newblock {\em Commun. {ACM}}, 28(2):202--208, 1985.

\end{thebibliography}

\newpage

\end{document}